\documentclass[runningheads]{llncs}
\usepackage{kotex}
\usepackage{color}
\usepackage{array}
\usepackage{umoline}
\usepackage{amssymb}
\usepackage{graphicx}
\usepackage{multirow}
\usepackage{url}
\usepackage{amsfonts}
\usepackage{listings}
\usepackage{algorithm}
\usepackage[noend]{algpseudocode}

\usepackage{amsmath}
\usepackage{caption}
\usepackage[misc]{ifsym}

\begin{document}

\title{Fast Cartesian Tree Matching}
%
%
\author{Siwoo Song\inst{1} \and Cheol Ryu\inst{1} \and Simone Faro\inst{2} \and Thierry Lecroq\inst{3} \and Kunsoo Park\inst{1}\textsuperscript{(\Letter)}}
\authorrunning{S. Song et al.}
%
\institute{Seoul National University, Seoul, Korea\\ \email{\{swsong,cryu,kpark\}@theory.snu.ac.kr} \and
University of Catania, Catania, Italy\\ \email{faro@dmi.unict.it} \and Normandie University, Rouen, France\\ \email{thierry.lecroq@univ-rouen.fr}}

%
\maketitle              
\begin{abstract}
Cartesian tree matching is the problem of finding all substrings of a given text which have the same Cartesian trees as that of a given pattern. So far there is one linear-time solution for Cartesian tree matching, which is based on the KMP algorithm. We improve the running time of the previous solution by introducing new representations. We present the framework of a binary filtration method and an efficient verification technique for Cartesian tree matching. Any exact string matching algorithm can be used as a filtration for Cartesian tree matching on our framework. We also present a SIMD solution for Cartesian tree matching suitable for short patterns.
By experiments we show that known string matching algorithms combined on our framework of binary filtration and efficient verification produce algorithms of good performances for Cartesian tree matching. 
\end{abstract}
\section{Introduction}
String matching is one of fundamental problems in computer science. There are generalized matchings such as parameterized matching \cite{PARA1,PARA2}, swapped matching \cite{SWAP1,SWAP2}, overlap matching \cite{Overlap}, jumbled matching \cite{JUMBLE}, and so on. These problems are characterized by the way of defining a match, which depends on the application domains of the problems. In particular, order-preserving matching \cite{OPM,OPM2,OPM3} and Cartesian tree matching \cite{CTM} deal with the order relations between numbers.

The Cartesian tree \cite{CT} is a tree data structure that represents a string, focusing on the orders between elements of the string. Park et al.~\cite{CTM} introduced a metric of match called Cartesian tree matching. It is the problem of finding all substrings of a text $T$ which have the same Cartesian trees as that of a pattern $P$.
Cartesian tree matching can be applied to finding patterns in time series data such as share prices in stock markets, like order-preserving matching, but sometimes it may be more appropriate as indicated in \cite{CTM}.
Fig.~\ref{fig:CTmatching} shows an example of Cartesian tree matching. Suppose $T=(10,12,16,15,6,14,9,12,11,14,9,17,12,10,12)$ and $P=(3,1,6,4,8,6,7,5,9)$. The Cartesian tree of substring $u=(15,6,14,9,12,11,14,9,17)$ is the same as that of $P$. Note that if we use order-preserving matching instead of Cartesian tree matching as a metric, $u$ does not match $P$. 
\begin{figure}[t]
\centering
\includegraphics[height=4.0cm]{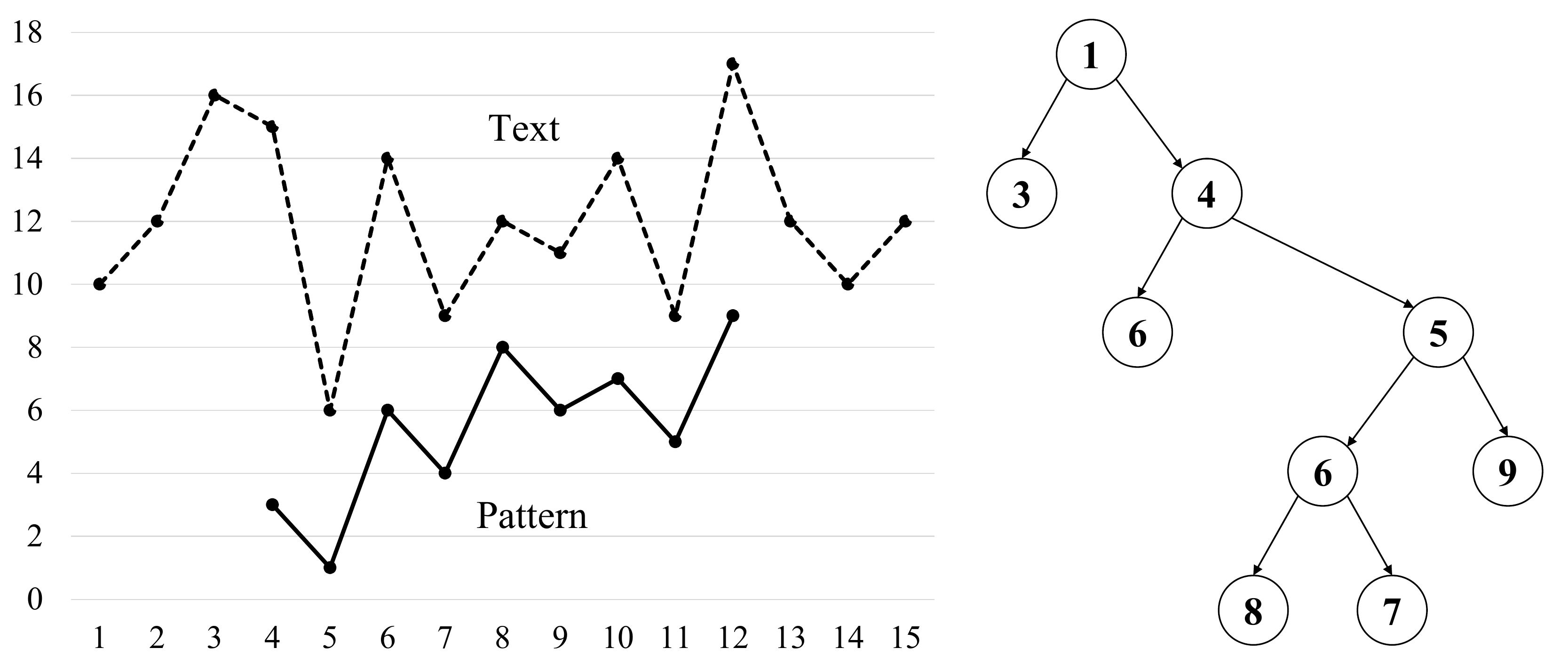}
\setlength{\belowcaptionskip}{-10pt} 
\caption{Cartesian tree matching, and Cartesian tree corresponding to pattern.}
\label{fig:CTmatching}
\end{figure}

String matching algorithms have been designed over the years. To speed up the search phase of string matching, algorithms based on automata and bit-parallelism were developed \cite{AOSO,SBNDM}. In recent years, the SIMD instruction set architecture gave rise to packed string matching, where one can compare packed data elements in parallel. 
In the last few years, many solutions for order-preserving matching have been proposed. Given a text of length $n$ and a pattern of length $m$, Kubica et al.~\cite{OPM3} and Kim et al.~\cite{OPM} gave $O(n + m \log m)$ time solutions based on the KMP algorithm. Cho et al.~\cite{CHO} presented an algorithm using the Boyer–Moore approach. Chhabra and Tarhio \cite{FilterOPM} presented a new practical solution based on filtration, and Chhabra et al.~\cite{SIMDOPM} gave a filtration algorithm using the Boyer-Moore-Horspool approach and SIMD instructions. Cantone et al.~\cite{OrderOPM} proposed filtration methods using the $q$-neighborhood representation and SIMD instructions. These filtration methods \cite{FilterOPM,SIMDOPM,OrderOPM} take sublinear time on average.

In this paper we introduce new representations, \emph{prefix-parent representation} and \emph{prefix-child representation}, which can be used to decide whether two strings have the same Cartesian trees or not. Using these representations, we improve the running time of the previous Cartesian tree matching algorithm in \cite{CTM}. We also present a binary filtration method for Cartesian tree matching, and give an efficient verification technique for Cartesian tree matching based on the \emph{global-parent representation}.
On the framework of our binary filtration method and efficient verification technique, we can apply any known string matching algorithm \cite{SkipSearch,HORSPOOL,SBNDM} as a filtration for Cartesian tree matching.
In addition, we present a SIMD solution for Cartesian tree matching based on the global-parent representation, which is suitable for short patterns.
We conduct experiments comparing many algorithms for Cartesian tree matching, which show that known string matching algorithms combined on the framework of our binary filtration and efficient verification for Cartesian tree matching produce algorithms of good performances for Cartesian tree matching.

This paper is organized as follows. In Section 2, we describe notations and the problem definition. In Section 3, we present an improved linear-time algorithm using new representations. In Section 4, we present the framework of binary filtration and efficient verification. In Section 5, we present a SIMD solution for short patterns. In Section 6, we give the experimental results of the previous algorithm and the proposed algorithms.

\section{Preliminaries}
\subsection{Basic notations}

A string is defined as a finite sequence of elements in an alphabet $\Sigma$. In this paper, we will assume that $\Sigma$ has a total order $<$. For a string $S$, $S[i]$ represents the $i$th element of $S$, and $S[i..j]$ represents a substring of $S$ from the $i$th element to the $j$th element. If $i>j$ then $S[i..j]$ is an empty string.

We will say $S[i]\prec S[j]$, if and only if $S[i]<S[j]$, or $S[i]$ and $S[j]$ have the same value with $i<j$.
Note that $S[i]=S[j]$ (as elements of the string) if and only if $i=j$. Unless stated otherwise, the minimum is defined by $\prec$.

\subsection{Cartesian tree matching}
A string $S$ can be associated with its corresponding Cartesian tree $CT(S)$ \cite{CT} according to the following rules:
\begin{itemize}
    \item If $S$ is an empty string, then $CT(S)$ is an empty tree.
    \item If $S[1..n]$ is not empty and $S[i]$ is the minimum value among $S$, then $CT(S)$ is the tree with $S[i]$ as the root, $CT(S[1..i-1])$ as the left subtree, and $CT(S[i+1..n])$ as the right subtree.
\end{itemize}
\textit{Cartesian tree matching} is to find all substrings of the text which have the same Cartesian trees as that of the pattern. Formally, Park et al.~\cite{CTM} define it as follows: 
\begin{definition}
\emph{(Cartesian tree matching) Given two strings text $T[1..n]$ and pattern $P[1..m]$, find every $1 \le i \le n-m+1$ such that $CT(T[i..i+m-1]) = CT(P[1..m])$.}
\end{definition}
Instead of building the Cartesian tree for every position in the text to solve Cartesian tree matching, Park et al.~\cite{CTM} use the following representation for a Cartesian tree.
\begin{definition}
\emph{(Parent-distance representation) Given a string $S[1..n]$, the parent-distance representation of $S$ is a function $\mathcal{PD}_S$, which is defined as follows:}
\begin{equation*}
    \mathcal{PD}_S(i) = \begin{cases} i - \max_{1 \leq j < i}\{{j : S[j] \prec S[i]} \} &\text{\emph{if such $j$ exists}}\\
    0 &\text{\emph{otherwise.}} \end{cases}
\end{equation*}
\end {definition}
Since the parent-distance representation has a one-to-one mapping to the Cartesian tree \cite{CTM}, it can replace the Cartesian tree without any loss of information.

\section {Fast linear Cartesian tree matching} \label{sec:linearalg}

The previous algorithm for Cartesian tree matching due to Park et al.~\cite{CTM} is based on the KMP algorithm \cite{KMP}. They changed the pattern and the text to parent-distance representations and found matches using the KMP algorithm. To compute the parent-distance representations of substrings of the text using $O(m)$ space, however, they used a deque data structure. We improve the text search phase of the previous algorithm by removing the overhead of computing parent-distance representations including deque operations.

In the text search phase of the previous algorithm, the parent-distance of each element $T[i]$ in $T[i-q..i]$ is computed to check whether it matches $\mathcal{PD}_P(q+1)$ when we know that $\mathcal{PD}_{P[1..q]}$ matches $\mathcal{PD}_{T[i-q..i-1]}$. We can do it directly without computing the parent-distances of text elements by using following representations: \emph{prefix-parent representation} and \emph{prefix-child representation}.

\begin{figure}[t]
\centering
\includegraphics[height=3.6cm]{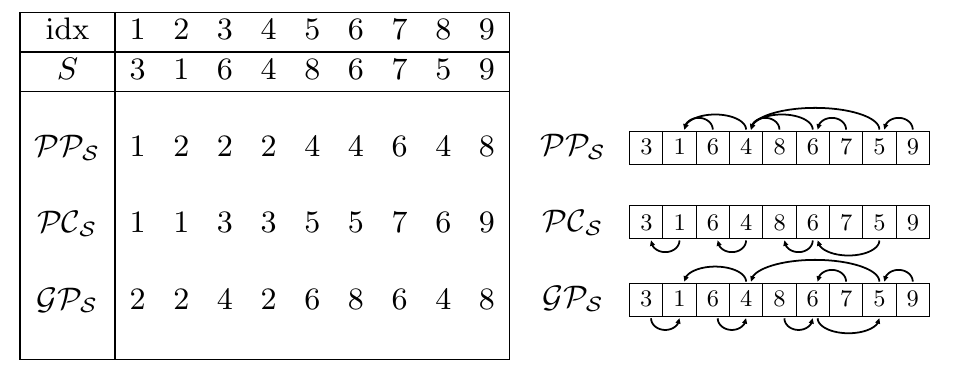}
\setlength{\belowcaptionskip}{-12pt} 
\caption{$\mathcal{PP}_S, \mathcal{PC}_S, \mathcal{GP}_S$ for $S=(3,1,6,4,8,6,7,5,9)$.}
\label{fig:pppc}
\end{figure}

\begin{definition}
\emph{(prefix-parent representation) Given a string $S[1..n]$, the prefix-parent representation of $S$ is a function $\mathcal{PP}_S$, which is defined as follows:}

\begin{equation*}
    \mathcal{PP}_S(i) = \begin{cases} \max_{1 \leq j < i}\{{j : S[j] \prec S[i]} \} &\text{\emph{if such $j$ exists}}\\
    i &\text{\emph{otherwise.}} \end{cases}
\end{equation*}
\end{definition}
Since $\mathcal{PP}_S(i) = i-\mathcal{PD}_S(i)$, the prefix-parent representation also has a one-to-one mapping to the Cartesian tree. 

\begin{definition}
\emph{(prefix-child representation) Given a string $S[1..n]$, the prefix-child representation of $S$ is a function $\mathcal{PC}_S$, which is defined as follows: $\mathcal{PC}_S(1)=1$, and for $i\geq 2$,}
\begin{equation*}
    \mathcal{PC}_S(i) = \begin{cases}
    j \text{\emph{ such that }} S[j] \text{\emph{ is minimum for }}1 \leq j < i  &\text{\emph{if }}\mathcal{PP}_S(i)=i\\
    i &\text{\emph{if }} \mathcal{PP}_S(i)=i-1\\
    j \text{\emph{ such that }} S[j] \text{\emph{ is minimum for }}\mathcal{PP}_S(i) < j < i  &\text{\emph{if }}\mathcal{PP}_S(i)<i-1\text{\emph{.}}
    \end{cases}
\end{equation*}
\end{definition}

In other words, $S[\mathcal{PC}_S(i)]$ is a child of $S[i]$, because $S[\mathcal{PC}_S(i)]$ is the root of $CT(S[\mathcal{PP}_S(i)+1..i-1])$ when $\mathcal{PP}_S(i)< i-1$, and $S[\mathcal{PC}_S(i)]$ is the root of $CT(S[1..i-1])$ when $\mathcal{PP}_S(i)=i$. When $\mathcal{PP}_S(i)=i-1$, there is no child of $S[i]$ in $CT(S[1..i])$, and thus we set $\mathcal{PC}_S(i)$ as $i$.

Fig.~\ref{fig:pppc} shows the prefix-parent representation (resp. the prefix-child representation) of string $S = (3, 1, 6, 4, 8, 6, 7, 5, 9)$ by arrows. The arrow starting from $S[i]$ indicates $\mathcal{PP}_S(i)$ (resp. $\mathcal{PC}_S(i)$). If $\mathcal{PP}_S(i)=i$ (resp. $\mathcal{PC}_S(i)=i$), we omit the arrow.


The advantage of using the prefix-child representation and the prefix-parent representation is that we can check whether each text element matches the corresponding pattern element in constant time without computing its parent-distance \cite{CTM}.

\begin{algorithm}[t]
\caption{Text search of Cartesian tree matching}\label{alg:textsearch}
\begin{algorithmic}[1]
\Procedure{CARTESIAN-TREE-MATCH}{$T[1..n], P[1..m]$}
    \State $(\mathcal{PP}_P,\mathcal{PC}_P)$ $\gets$ PREFIX-PARENT-CHILD-REP($P$)
    \State $\pi$ $\gets$ FAILURE-FUNC($P$)
    \State $q \gets 0$
    
    \For{$i \gets 1$ \textbf{to} $n$}
        \While{$q \neq 0$}
            \If{$T[i-q-1+\mathcal{PP}_P(q+1)]\preceq T[i]\preceq T[i-q-1+\mathcal{PC}_P(q+1)]$}
                \State \textbf{break}
            \Else
                \State $q \gets \pi[q]$
            \EndIf
        \EndWhile
        
        \State $q \gets q + 1$
        
        \If{$q = m$}
            \State \textbf{print} ``Match occurred at $i-m+1$"
            \State $q \gets \pi[q]$
        \EndIf
    \EndFor
\EndProcedure
\end{algorithmic}
\end{algorithm}

\begin {theorem}
\label{thm:pppc}
\emph{Given two strings $P$ and $S$, assume that $P[1..q]$ and $S[1..q]$ have the same prefix-parent representations.
If $S[\mathcal{PP}_P(q+1)] \preceq S[q+1] \preceq S[\mathcal{PC}_P(q+1)]$, then $P[1..q+1]$ and $S[1..q+1]$ have the same prefix-parent representations, and vice versa.}
\end {theorem}

\begin{proof}

$(\Longrightarrow)$
If $q=0$, $P[1]$ and $S[1]$ always have the same prefix-parent $1$.
Now let's assume $q\geq 1$. There are three cases, in each of which we show that $\mathcal{PP}_P(q+1)=\mathcal{PP}_S(q+1)$.
\begin{enumerate}
\item Case $\mathcal{PP}_P(q+1)=q+1$: Since $P[\mathcal{PC}_P(q+1)]$ is the minimum element in $P[1..q]$ and $\mathcal{PP}_P(i) = \mathcal{PP}_S(i)$ for $1\le i\le q$, $S[\mathcal{PC}_P(q+1)]$ is also the minimum element in $S[1..q]$. Therefore, if $S[q+1]\preceq S[\mathcal{PC}_P(q+1)]$ holds, then we have $\mathcal{PP}_S(q+1)=q+1$.
\item Case $\mathcal{PP}_P(q+1)=q$: Since $S[q]\preceq S[q+1]$, we have $\mathcal{PP}_S(q+1)=q$.
\item Case $\mathcal{PP}_P(q+1)<q$: Since $P[\mathcal{PC}_P(q+1)]$ is the minimum element in $P[\mathcal{PP}_P(q+1)+1..q]$ and $\mathcal{PP}_P(i) = \mathcal{PP}_S(i)$ for $1\le i\le q$, $S[\mathcal{PC}_P(q+1)]$ is also the minimum element in $S[\mathcal{PP}_P(q+1)+1..q]$. Therefore, if $S[\mathcal{PP}_P(q+1)]\preceq S[q+1]\preceq S[\mathcal{PC}_P(q+1)]$ holds, then $\mathcal{PP}_S(q+1)=\mathcal{PP}_P(q+1)$.
\end{enumerate}

$(\Longleftarrow)$
It is trivial by definitions of $\mathcal{PP}$ and $\mathcal{PC}$. \qed
\end {proof}

With the prefix-parent representation and the prefix-child representation of pattern $P$, we can simplify the text search. For each element $T[i]$, we can check $\mathcal{PP}_P(q+1) = \mathcal{PP}_{T[i-q..i]}(q+1)$ by comparing $T[i]$ with the elements in $T[i-q..i]$ whose indices correspond to $\mathcal{PP}_P(q+1)$ and $\mathcal{PC}_P(q+1)$ in $P$. Using this idea, we don't have to compute $\mathcal{PP}_{T[i-q..i]}(q+1)$. Algorithm~\ref{alg:textsearch} describes the algorithm to do this. We compute the failure function $\pi$ in the same way as \cite{CTM} does.

Given a string $P[1..m]$, we can compute the prefix-child representation and the prefix-parent representation simultaneously in linear time using a stack. $\mathcal{PP}_P(i)=j$ means that $P[j] \prec P[k]$ for $j<k<i$. The same is true for $\mathcal{PC}_P(i)$. On the stack, therefore, we maintain only $j$'s which satisfy $P[j] \prec P[k]$ for $j<k<i$ while scanning from $i=1$ to $m$. 
Suppose that $j_1,j_2,\dots,j_r$ are on the stack when we are computing $\mathcal{PP}_P(i)$ and $\mathcal{PC}_P(i)$. (We assume that $j_{r+1}=i$.)
Then, $(P[j_1],P[j_2],\dots,P[j_r])$ forms an increasing subsequence of $P$. When we consider a new index $i$, we pop the indices $j_r, j_{r-1},\dots,j_{t+1}$ repeatedly until we have $P[j_{t}]\prec P[i]$. If there exists such an index $j_{t}$, we set $\mathcal{PP}_P(i)=j_{t}$ and $\mathcal{PC}_P(i)=j_{t+1}$.
(If $t=r$, then $\mathcal{PC}_P(i)=j_{t+1}=i$.)
Otherwise, $P[i]$ is the minimum element in $P[1..i]$, and thus $\mathcal{PP}_P(i)=i$ and $\mathcal{PC}_P(i)=j_1$. 
Finally, we push $i$ onto the stack.
Algorithm~\ref{alg:computepdcdref} describes the algorithm to compute $\mathcal{PP}_P$ and $\mathcal{PC}_P$ simultaneously.

\begin{algorithm}[t]
\caption{Computing prefix-parent and prefix-child representations}\label{alg:computepdcdref}
\begin{algorithmic}[1]
\Procedure{PREFIX-PARENT-CHILD-REP}{$P[1..m]$}
    \State $ST \gets \text{an empty stack}$ 
    \For{$i \gets 1$ \textbf{to} $m$}
        \State $j_{next}\gets i$
        \While{$ST$ is not empty}
            \State $j \gets ST.top$
            \If{$P[j] \prec P[i]$}
                \State \textbf{break}
            \EndIf
            \State $ST.pop$
            \State $j_{next}\gets j$
        \EndWhile
        \State $\mathcal{PC}_P(i)\gets j_{next}$
        \If{$ST$ is empty}
            \State $\mathcal{PP}_P(i) \gets i$
        \Else
            \State $\mathcal{PP}_P(i) \gets j$
        \EndIf

        \State $ST.push(i)$
    \EndFor
    \State \textbf{return} $(\mathcal{PP}_P, \mathcal{PC}_P)$
\EndProcedure
\end{algorithmic}
\end{algorithm}


\section{Fast Cartesian tree matching with filtration}

In this section we present a practical solution based on filtration. Our solution for Cartesian tree matching consists of two phases: filtration and verification. First, the text is filtered with some exact string matching algorithm using a binary representation. In the second phase, the potential candidates are verified using a global-parent representation.

\subsection{Filtration} \label{sub:filtration}
In the filtration phase, a string $S$ is translated into a \emph{binary representation} $\beta_S$ as follows.
\begin{definition}
\emph{(binary representation) Given a string $S[1..n]$, the binary representation of $S$ is a binary string $\beta_S$ of length $n-1$, which is defined as follows:} 
\begin{equation*}
    \beta_S[i] = \begin{cases} 0 &\text{\emph{if }} \mathcal{PP}_S(i+1) = i \\
    1 &\text{\emph{otherwise,}} \end{cases}
\end{equation*}
\emph{for each }$1 \le i \le n-1$.
\end{definition}
One can easily check whether $\mathcal{PP}_S(i+1)=i$ is true or not by comparing $S[i]$ and $S[i+1]$: $\mathcal{PP}_S(i+1)=i$ if and only if $S[i]\prec S[i+1]$.
The following theorem proves that the binary representation can be used to filter a text $T$ to search for all Cartesian tree matching occurrences of a pattern $P$.

\begin{theorem}\label{binary}
\emph{Let $P$ and $T$ be two strings of lengths $m$ and $n$, respectively, and let $\beta_P$ and $\beta_T$ be the binary representations associated with $P$ and $T$, respectively. If $CT(P[1..m]) = CT(T[i..i+m-1])$, then $\beta_P[j] = \beta_T[i+j-1]$ for $1 \le j \le m-1$.} 
\end{theorem}

\begin{proof}
The prefix-parent representation has a one-to-one mapping to the Cartesian tree. Therefore, if $CT(P[1..m]) = CT(T[i..i+m-1])$, then $\mathcal{PP}_P(j+1)=\mathcal{PP}_T(i+j)$  for $0 \le j \le m-1$. If $\mathcal{PP}_P(j+1)=\mathcal{PP}_T(i+j)$, then $\beta_P[j] = \beta_T[i+j-1]$ for $1 \le j \le m-1$.
\end{proof}
Theorem~\ref{binary} guarantees that any standard exact string matching algorithm can be used as a filtration procedure. As the exact string matching algorithm returns matches of $\beta_P$ in $\beta_T$, these matches are only possible candidates of Cartesian tree matching which should be verified.

Cantone et al.~\cite{OrderOPM} presented two filtration methods other than the binary representation to solve order-preserving matching. They used the property that $T$ doesn't match $P$ at position $i$ if there are two positions $j$ and $k$ such that $P[j]\preceq P[k]\Leftrightarrow T[i+j-1]\preceq T[i+k-1]$ doesn't hold. Thus any comparison result between two positions can be used for filtration. In Cartesian tree matching, however, even if there exist such $j$ and $k$, the corresponding Cartesian trees can be the same when $|j-k|> 1$. Therefore, we cannot use these filtration methods for Cartesian tree matching.

\subsection{Verification}\label{sub:verification}
In the verification phase, we have to check whether the candidates found by the filtration phase are actual matches or not.
This checking can be done using prefix-parent and prefix-child representations by Theorem~\ref{thm:pppc}, which takes 2 comparisons per element. In order to reduce the number of comparisons to 1, we introduce another representation as follows.

\begin{definition}
\emph{(Global-parent representation) Given a string $S[1..n]$, the global-parent representation of $S$ is a function $\mathcal{GP}_S$, which is defined as follows:}

\begin{equation*}
    \mathcal{GP}_S(i) = \begin{cases} j & \text{\emph{such that }} \mathcal{PC}_S(j)=i\text{\emph{ for }}j>i\\
    \mathcal{PP}_S(i) & \text{\emph{if such $j$ doesn't exist.}}
    \end{cases}
\end{equation*}
\end {definition}
$\mathcal{GP}_S(i)$ is well-defined because there is at most one $j>i$ which satisfies $\mathcal{PC}_S(j)=i$.
Fig.~\ref{fig:pppc} shows the global-parent representation by arrows. The arrow starting from $S[i]$ indicates the global parent of $S[i]$. If $\mathcal{GP}_S(i)=i$, we omit the arrow.

\begin {theorem}
\label{thm:gp}
\emph{Two strings $P[1..m]$ and $S[1..m]$ have the same Cartesian trees if and only if $S[\mathcal{GP}_P(i)] \preceq S[i]$ for all $1\le i\le m$.}
\end {theorem}

\begin{proof}
We will prove that $S[\mathcal{PP}_P(i)] \preceq S[i] \preceq S[\mathcal{PC}_P(i)]$ for all $1\le i\le m$ if and only if $S[\mathcal{GP}_P(i)] \preceq S[i]$ for all $1\le i\le m$.

$(\Longrightarrow)$
It is trivial by definition of $\mathcal{GP}$.

$(\Longleftarrow)$ Assume $S[\mathcal{GP}_P(i)] \preceq S[i]$ for all $1\le i\le m$. For any $1\le k\le m$, we first show $S[k] \preceq S[\mathcal{PC}_P(k)]$, and then we show $S[\mathcal{PP}_P(k)] \preceq S[k]$.
\begin{enumerate}
\item (Proof of $S[k]\preceq S[\mathcal{PC}_P(k)]$)
There are two cases: $\mathcal{PC}_P(k)=k$ and $\mathcal{PC}_P(k)\neq k$. If $\mathcal{PC}_P(k)=k$, then $S[k]\preceq S[\mathcal{PC}_P(k)]$ holds trivially. Otherwise, since $\mathcal{GP}_P({\mathcal{PC}_P(k)})=k$, $S[k]=S[\mathcal{GP}_P({\mathcal{PC}_P(k)})]\preceq S[\mathcal{PC}_P(k)]$. Therefore, $S[k] \preceq S[\mathcal{PC}_P(k)]$ holds.

\item (Proof of $S[\mathcal{PP}_P(k)] \preceq S[k]$)
If $\mathcal{GP}_P(k)=\mathcal{PP}_P(k)$, then $S[\mathcal{PP}_P(k)]=S[\mathcal{GP}_P(k)] \preceq S[k]$. So we only have to consider the case that there is $k_1>k$ which satisfies $\mathcal{PC}_P(k_1)=k$.
Let $k=k_0< k_1<\dots<k_r \le m$ be a sequence such that $\mathcal{PC}_P(k_{l+1})=k_l$, and there is no $k_{r+1}> k_r$ which satisfies $\mathcal{PC}_P(k_{r+1})=k_r$. Since $(k_0,k_1,\dots,k_r)$ is a strictly increasing sequence, such $k_r$ always exists. Note that $\mathcal{GP}_P(k_l)=k_{l+1}$ except for $\mathcal{GP}_P(k_r)$. On the sequence, there may or may not exist $j$ such that $\mathcal{PP}_P(k_j)=k_j$.

Suppose that there exists some $j$ such that $\mathcal{PP}_P(k_j)=k_j$. Since $k_{j-1}=\mathcal{PC}_P(k_j)$, $P[k_{j-1}]$ is the minimum element in $P[1..k_j-1]$, and so $\mathcal{PP}_P(k_{j-1})=k_{j-1}$. Proceeding inductively, $\mathcal{PP}_P(k_l)=k_l$ for all $l\le j$. Thus $S[\mathcal{PP}_P(k)] \preceq S[k]$ holds trivially.

Now we consider the case that $\mathcal{PP}_P(k_j)\neq k_j$ for all $j$. Then, we have $S[k_0]\succeq S[k_1]\succeq \dots \succeq S[k_r]\succeq S[\mathcal{GP}_P(k_r)]=S[\mathcal{PP}_P(k_r)]$ by the assumption that $S[\mathcal{GP}_P(i)]\preceq S[i]$ for all $i$. 
We now show $\mathcal{PP}_P(k_r)=\mathcal{PP}_P(k)$ as follows.
Since $\mathcal{PC}_P(k_r)=k_{r-1}$, $P[k_{r-1}]$ is the minimum element in $P[\mathcal{PP}_P(k_r)+1..k_r-1]$, and $P[k_{r-1}]\succeq P[\mathcal{PP}_P(k_r)]$. Hence, we have $\mathcal{PP}_P({k_{r-1}})=\mathcal{PP}_P(k_r)$. Inductively, we can show that $\mathcal{PP}_P({k_{0}})=\mathcal{PP}_P(k_{1})=\dots =\mathcal{PP}_P(k_r)$.
Therefore, $S[\mathcal{PP}_P(k)] \preceq S[k]$ holds. $\qed$
\end{enumerate} 
\end{proof}

By Theorem~\ref{thm:gp}, we only have to compare once for each element in the verification phase.
For a potential candidate obtained from the filtration phase (say, it starts from $T[i]$), we compare $T[i+q-1]$ and $T[i+\mathcal{GP}_P(q)-1]$ from $q=1$ to $m$. The candidate is discarded when there exists $q$ such that $T[i+q-1]\prec T[i+\mathcal{GP}_P(q)-1]$.

We compute the global-parent representation using a stack, as in computing the prefix-parent and the prefix-child representations. The only difference is that first we set $\mathcal{GP}_P(i)$ as $\mathcal{PP}_P(i)$, and then if we find $j$ such that $\mathcal{PC}_P(j)=i$ we update $\mathcal{GP}_P(i)$ to $j$. 
%
%

\subsection{Sublinear time on average}
The proof of sublinearity is similar to the analysis of order-preserving matching with filtration \cite{FilterOPM}. Let's assume that the elements in the pattern $P$ and the text $T$ are independent of each other and the distribution is uniform. The verification phase takes time proportional to the pattern length times the number of potential candidates. When alphabet size is $|\Sigma|$, the probability that $\beta_P[i]=0$ (i.e., probability that $P[i]\prec P[i+1]$) is $(|\Sigma|^2+|\Sigma|)/(2|\Sigma|^2)$, since there are $|\Sigma|^2$ pairs and $|\Sigma|$ pairs among them have equal elements. Similarly, the probability that $\beta_P[i]=1$ is $(|\Sigma|^2-|\Sigma|)/(2|\Sigma|^2)$, and it is the same for $\beta_T[i]$. Therefore, the probability that $\beta_P[i]=\beta_T[i]$ is $((|\Sigma|^2+|\Sigma|)/(2|\Sigma|^2))^2+((|\Sigma|^2-|\Sigma|)/(2|\Sigma|^2))^2=1/2+1/(2|\Sigma|^2)$. As the pattern length increases, the number of potential candidates decreases exponentially, and the verification time approaches zero. Hence, the filtration time dominates. So if the filtration method takes a sublinear time in the average case, the total algorithm takes a sublinear time in the average case, too.
  
\subsection{SIMD instructions}
 
When we use the Boyer-Moore-Horspool algorithm \cite{HORSPOOL} and the Alpha skip search algorithm \cite{SkipSearch} as the filtration method, we pack four 32-bit numbers or sixteen 8-bit numbers into a register, as in order-preserving matching algorithms \cite{SIMDOPM,OrderOPM}. Each pair of two corresponding packed data elements can be compared in parallel using streaming SIMD extensions (SSE) \cite{SSE}.  In the case of 32-bit integers, for example, we compute $(T[i+3]>T[i+4])$, $(T[i+2]>T[i+3])$, $(T[i+1]>T[i+2])$, and $(T[i]>T[i+1])$ in parallel as in Algorithm \ref{alg:simdInt}, where
instruction \_mm\_loadu\_si128((\_\_m128i *)($T+i$)) loads four 32-bit integers from memory $T+i$ into a 128-bit register, instruction \_mm\_cmpgt\_epi32($a$, $b$)  compares four pairs of packed 32-bit integers and returns the results of the comparisons into a 128-bit register, instruction \_mm\_castsi128\_ps casts the integer type to the float type, and instruction \_mm\_movemask\_ps selects only the most significant bits of the 4 floats. 
Comparing a pair of sixteen 8-bit numbers can be done similarly.

\setlength{\textfloatsep}{12pt}
\begin{algorithm}[t]
\caption{Compare integers in parallel}\label{alg:simdInt}
\begin{algorithmic}[1]
\Procedure{CompareUsingSIMD}{$T[1..n], i$}
    \State \_\_m128i $a$ $\gets$ \_mm\_loadu\_si128((\_\_m128i *)($T+i$))
    \State \_\_m128i $b$ $\gets$ \_mm\_loadu\_si128((\_\_m128i *)($T+i+1$))
    \State \_\_m128i $r$ $\gets$ \_mm\_cmpgt\_epi32($a$, $b$)
    \State \textbf{return} \_mm\_movemask\_ps(\_mm\_castsi128\_ps($r$)) 
\EndProcedure
\end{algorithmic}
\end{algorithm}

\section{SIMD solution for short patterns} \label{pmct}
In this section we present an algorithm that works when the alphabet consists of 1-byte characters and the pattern length $m$ is at most 16. As shown in Section~\ref{sub:verification}, we test $T[s+i-1]\succeq T[s+\mathcal{GP}_P(i)-1]$ for $1\le i\le m$ to check for an occurrence at position $s$ of the text $T$.

Let $W$ be a word of 16 bytes containing the current window of the text, i.e., $W=T[s..s+15]$. For $1\le i\le m$, we define $W_i$ (word obtained from $W$ by shifting $i-\mathcal{GP}_P(i)$ positions to the left or to the right, depending on the sign of $i-\mathcal{GP}_P(i)$) as follows:
\begin{equation*}
W_i=\begin{cases}W\ll (\mathcal{GP}_P(i)-i) &\text{ if } i<\mathcal{GP}_P(i)\\
W\gg (i-\mathcal{GP}_P(i)) &\text{ if } i>\mathcal{GP}_P(i).
\end{cases}
\end{equation*}
For fixed $i$, we can find the positions $j$ which satisfy $W[j+i-1]\succeq W[j+\mathcal{GP}_P(i)-1]$ for $0\le j\le 15$ in parallel by comparing $W_i$ to $W$ using SIMD instructions. The satisfying positions for all $1\le i\le m$ are the occurrences of the pattern. 
The details of the algorithm are as follows. 
We test whether $W_i[j]\preceq W[j]$ for $0\le j\le15$ in parallel using the SIMD instruction $R_i=\text{\_mm\_cmpgt\_epi8}(W,W_i)$ for $i<\mathcal{GP}_P(i)$ or $R_i={\sim}\text{\_mm\_cmpgt\_epi8}(W_i,W)$ for $i>\mathcal{GP}_P(i)$. 
(In order to get only significant bits when computing $R_i$, we use instruction \_mm\_movemask\_epi8.)
Then we compute $q=\text{AND}_{i=1}^{m}(R_i\ll (i-1))$.
Finally, we report a match at position $s+j$ of the text if $q[j]=1$.

\begin{example}
Let's consider an example of the pattern $P = (3, 1, 6, 4, 8)$ and the window of the text $W = (10, 12, 16, 15, 6, 14, 9, 12, 11, 14, 9, 17, 12, 13, 12, 10)$. We observe that since $1-\mathcal{GP}_P(1)=3-\mathcal{GP}_P(3)$, $R_1=R_3$. Moreover we do not need to compute $R_2$, since $2-\mathcal{GP}_P(2)=0$. Hence we compute $R_1$, $R_4$, and $R_5$.
\vspace{2mm}

\setlength{\tabcolsep}{2.8pt}
\begin{tabular}{l c r r r r r r r r r r r r r r r r}
$W$&=&10,&12,&16,&15,&06,&14,&09,&12,&11,&14,&09,&17,&12,&13,&12,&10\\
$W_1$&=&12,&16,&15,&06,&14,&09,&12,&11,&14,&09,&17,&12,&13,&12,&10\;&\\
$R_1$&=&0,&0,&1,&1,&0,&1,&0,&1,&0,&1,&0,&1,&0,&1,&1,&-\\
\end{tabular}

\setlength{\tabcolsep}{2.8pt}
\begin{tabular}{l c r r r r r r r r r r r r r r r r}
$W$&=&10,&12,&16,&15,&06,&14,&09,&12,&11,&14,&09,&17,&12,&13,&12,&10\\
$W_4$&=& &  &10,&12,&16,&15,&06,&14,&09,&12,&11,&14,&09,&17,&12,&13\\
$R_4$&=&-,&-,&1,&1,&0,&0,&1,&0,&1,&1,&0,&1,&1,&0,&1,&1\\
\end{tabular}

\setlength{\tabcolsep}{2.8pt}
\begin{tabular}{l c r r r r r r r r r r r r r r r r}
$W$&=&10,&12,&16,&15,&06,&14,&09,&12,&11,&14,&09,&17,&12,&13,&12,&10\\
$W_5$&=& &10,&12,&16,&15,&06,&14,&09,&12,&11,&14,&09,&17,&12,&13,&12\\
$R_5$&=&-,&1,&1,&0,&0,&1,&0,&1,&0,&1,&0,&1,&0,&1,&0,&0\\
\end{tabular}
\vspace{2mm}

The final result $q$ can be computed as follows:
\vspace{2mm}

\setlength{\tabcolsep}{4.5pt}
\begin{tabular}{l c r r r r r r r r r r r r r r r r}
$R_1$     &=&0,&0,&1,&1,&0,&1,&0,&1,&0,&1,&0,&1,&0,&1,&1,&-\\
$R_3\ll 2$&=&1,&1,&0,&1,&0,&1,&0,&1,&0,&1,&0,&1,&1,&-,&0,&0\\
$R_4\ll 3$&=&1,&0,&0,&1,&0,&1,&1,&0,&1,&1,&0,&1,&1,&0,&0,&0\\
$R_5\ll 4$&=&0,&1,&0,&1,&0,&1,&0,&1,&0,&1,&0,&0,&0,&0,&0,&0\\
\hline
$q$         &=&0,&0,&0,&1,&0,&1,&0,&0,&0,&1,&0,&0,&0,&0,&0, &0\\
\end{tabular} 
\end{example}

Therefore, we can report 3 matches. After we have tested a window of the text, we shift the current window to the right by $17-m$ positions. This algorithm takes $O(mn/(17-m))$ SIMD instructions. 

\section{Experiments}

\begin{table}[t]
\begin{center}
\renewcommand{\arraystretch}{1}
\begin{tabular}{|c | c | c | c | c c c | c c c c | c c c c | c|}
\hline
\multirow{2}{*}{Dataset}& \multirow{2}{*}{$m$} & KMP & IKMP & \multicolumn{3}{c|}{SBNDMCT} & \multicolumn{4}{c|}{BMHCT} & \multicolumn{4}{c|}{SKSCT} & PM \\
& & CT & CT & 2 & 4 & 6 & 4 & 8 & 12 & 16 & 4 & 8 & 12 & 16 & CT \\
\hline
Random & 5 & 10.52 & 6.84 & 4.99 & 4.42 & & 4.17 & & & & \textbf{3.31} & & & & \\
int & 9 & 10.71 & 6.83 & 2.71 & 2.31 & 1.95 & 1.95 & \textbf{1.64} & & & 1.91 & 2.26 & & & \\
& 17 & 10.69 & 6.83 & 1.39 & 1.34 & 0.95 & 1.31 & 0.80 & 0.86 & 1.60 & 1.13 & \textbf{0.45} & 0.61 & 3.91 & \\
& 33 & 10.69 & 6.83 & 0.72 & 0.70 & 0.65 & 1.07 & 0.51 & 0.51 & 1.01 & 0.76 & 0.32 & \textbf{0.30} & 0.48 & \\
& 65 & 10.71 & 6.83 & 0.72 & 0.71 & 0.66 & 0.98 & 0.44 & 0.43 & 0.71 & 0.61 & 0.27 & \textbf{0.24} & 0.28 & \\
\hline
Seoul & 5 & 5.08 & 3.07 & 2.67 & 2.91 & & 2.52 & & & & \textbf{2.27} & & & & \\
temp & 9 & 5.11 & 3.14 & 1.56 & 1.45 & 1.55 & 1.55 & \textbf{1.23} & & & 1.27 & 1.77 & & & \\
& 17 & 5.51 & 3.12 & 0.89 & 0.81 & 0.71 & 1.10 & 0.62 & 0.63 & 0.84 & 0.88 & \textbf{0.44} & 0.49 & 2.55 & \\
& 33 & 5.56 & 3.12 & 0.49 & 0.48 & 0.45 & 0.84 & 0.40 & 0.34 & 0.41 & 0.68 & 0.32 & \textbf{0.20} & 0.25 & \\
& 65 & 5.52 & 3.11 & 0.48 & 0.48 & 0.46 & 0.77 & 0.26 & 0.19 & 0.28 & 0.57 & 0.25 & 0.13 & \textbf{0.12} & \\
\hline
Random & 5 & 10.24 & 6.86 & 4.80 & 4.44 & & 3.95 & & & & 3.22 & & & & \textbf{0.50} \\
char & 7 & 10.32 & 6.86 & 3.53 & 2.89 & 4.47 & 2.39 & & & & 2.40 & & & & \textbf{0.84}\\
& 9 & 10.34 & 6.85 & 2.65 & 2.32 & 1.94 & 1.74 & \textbf{1.24} & & & 1.91 & 1.47 & & & 1.32\\
& 13 & 10.32 & 6.85 & 1.75 & 1.68 & 1.10 & 1.23 & 0.70 & 0.68 & & 1.34 & \textbf{0.45} & 1.15 & & 3.76\\
& 17 & 10.35 & 6.86 & 1.28 & 1.25 & 0.87 & 1.04 & 0.52 & 0.49 & 0.79 & 1.04 & \textbf{0.27} & 0.32 & 1.64 & \\
& 33 & 10.34 & 6.85 & 0.61 & 0.60 & 0.54 & 0.78 & 0.29 & 0.26 & 0.43 & 0.66 & 0.16 & \textbf{0.09} & 0.11 & \\
& 65 & 10.36 & 6.86 & 0.63 & 0.63 & 0.55 & 0.74 & 0.20 & 0.17 & 0.27 & 0.47 & 0.13 & \textbf{0.04} & 0.05 & \\
\hline
\end{tabular}
\end{center}
\setlength{\belowcaptionskip}{-2pt} 
\caption{Execution times in seconds for random patterns in texts (Random datasets: for 100 patterns, Seoul temperatures dataset: for 1000 patterns).}
\label{table:time}
\end{table}

\begin{figure}[t]
\centering
\includegraphics[height=5.0cm]{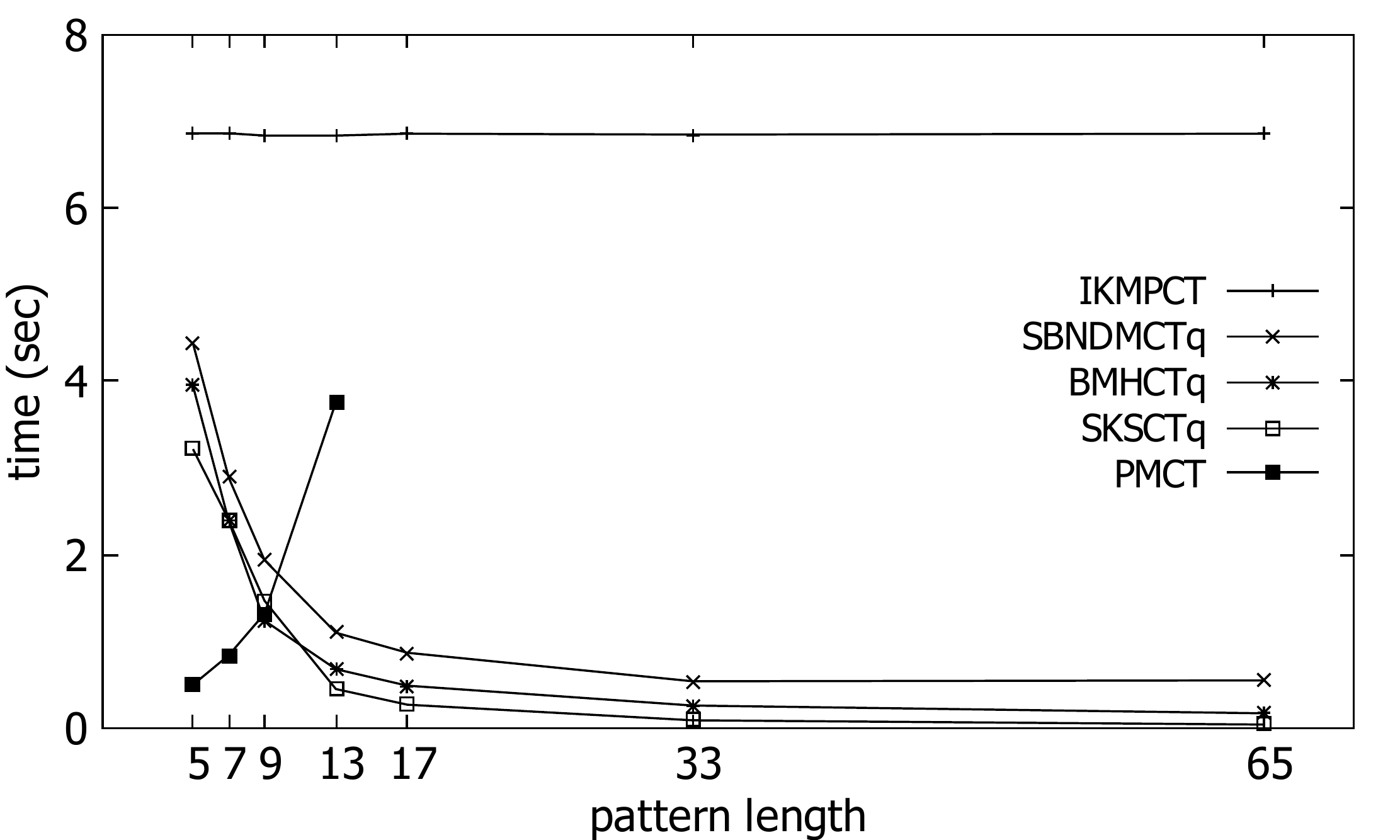}
\setlength{\belowcaptionskip}{3pt}
\caption{Execution times for the random character dataset.}
\label{fig:integer}
\end{figure}

In this section we conduct experiments comparing the following algorithms. 
\begin{itemize}
    \item KMPCT: algorithm of Park, Amir, Landau, and Park \cite{CTM}
    \item IKMPCT: our improved linear-time algorithm based on prefix-parent and prefix-child representations (Section~\ref{sec:linearalg})
    \item PMCT: SIMD solution for short patterns (Section~\ref{pmct})
    \item SBNDMCT$q$: algorithm based on the SBNDM$q$ filtration implemented by Faro and Lecroq \cite{SMART} on the binary representations of the text and the pattern (Section~\ref{sub:filtration}) and verification using the global-parent representation (Section~\ref{sub:verification}) \cite{SBNDM} (The following algorithms have the same framework as SBNDMCT$q$; only SBNDM$q$ is replaced by another filtration method.)
    \item BMHCT$q$: algorithm based on the $q$-gram Boyer-Moore-Horspool filtration using SIMD instructions \cite{HORSPOOL,QGRAM,SIMDOPM} 
    \item SKSCT$q$: algorithm based on the $q$-gram Alpha skip search filtration using SIMD instructions \cite{SkipSearch,OrderOPM}
\end{itemize}

We tested for two random datasets and one real dataset, which is a time series of Seoul temperatures. The first random dataset consists of 10,000,000 random integers. The second random dataset consists of 10,000,000 random characters. The Seoul temperatures dataset consists of 658,795 integers referring to the hourly temperatures in Seoul (multiplied by ten) in the years 1907-2019. In general, temperatures rise during the day and fall at night. Therefore, the Seoul temperatures dataset has more matches than random datasets. We picked 100 random patterns per pattern length from random datasets and 1000 random patterns per pattern length for the Seoul temperatures dataset. 

The experimental environments and parameters are as follows. All algorithms were implemented in C++11 and compiled with GNU C++ compiler version 4.8.5, and O3 and msse4 options were used. The experiments were performed on a CentOS Linux 7 with 128GB RAM and Intel Xeon CPU E5-2630 processor.

Table~\ref{table:time} shows the total execution times of Cartesian tree matching algorithms for random patterns (including the preprocessing). The best results are boldfaced. We choose the best results of the random character dataset from each algorithm regardless of $q$ and present them in Fig.~\ref{fig:integer} (except KMPCT because of readability). Our linear-time algorithm IKMPCT improves upon algorithm KMPCT of \cite{CTM} by about 35\%. In the random character dataset, PMCT is the fastest algorithm for short patterns. However, as the pattern length grows, algorithms based on the filtration method are much faster in practice. It can be seen that SKSCT is the fastest algorithm in most cases. When the pattern length is equal to 9, BMHCT utilizing 8-grams is the fastest algorithm, irrespective of the datasets. As pattern length grows, SKSCT utilizing 12-grams becomes the fastest algorithm.

Regardless of the data type, the results are almost consistent. In details, however, there are several differences. First, filtration algorithms, especially SKSCT algorithms, are slower at the Seoul temperatures dataset relatively. It's because there are more matches in the Seoul temperatures dataset. Second, when $q$ is large, BMHCT and SKSCT algorithms are faster in the random character dataset than in the random integer dataset. It's because the maximum number that we can compute in parallel is 16 in the character dataset while it is 4 in the integer dataset.

\vspace{\baselineskip}

\noindent\textbf{Acknowledgments.} Song, Ryu and Park were supported by Collaborative Genome Program for Fostering New Post-Genome industry  through the National Research Foundation of Korea(NRF) funded by the Ministry of  Science ICT and Future Planning (No. NRF-2014M3C9A3063541).

%
%

%
%
%
\bibliographystyle{splncs04}
\bibliography{mybibliography}

\end{document}